\documentclass [11pt]{article}

\usepackage{amsmath,amsthm,amscd,amssymb}

\usepackage[small, centerlast, it]{caption}
\usepackage{epsfig}
\usepackage[titles]{tocloft}
\usepackage[latin1]{inputenc}
\usepackage{verbatim} %for comments via \begin{comment}
\usepackage[active]{srcltx} %inverse search
\usepackage{fancyhdr}
\usepackage{caption}

\def\b1{{1\!\!1}}

\def\cC{{\ca C}}

\def\cG{{\ca G}}

\def\cL{{\ca L}}
\def\cM{{\ca M}}
\def\cV{{\ca V}}

\def\cO{{\ca O}}

\def\cR{{\ca R}}
\def\cS{{\ca S}}

\def\sH{{\mathsf H}}
\def\sP{{\mathsf P}}

\def\bN{{\mathbb N}}

\def\bR{{\mathbb R}}

\def\bS{{\mathbb S}}

\def\mL{\mathcal L}

\def\gB{{\mathfrak B}}

\def\gD{{\mathfrak D}}

\def\gu{{\mathfrak u}}

\def\beq{\begin{eqnarray}}
\def\eeq{\end{eqnarray}}

\newcommand{\ca}[1]{{\cal #1}}         %%  calligraphic

\usepackage{sectsty}
\sectionfont{\normalsize}%\normalfont
\subsectionfont{\normalsize\normalfont\itshape}
\newtheoremstyle{thm}
{12pt}% space above
{12pt}% space below
{\itshape}% body font
{}% h indent amount
{\itshape\bfseries}% theorem head font
{}% punctuation after theorem head
{1em}% space after theorem head
{}% theorem head spec (can be left empty, meaning `normal')

\theoremstyle{thm}
\newtheorem{theorem}{Theorem}

\newtheorem{proposition}[theorem]{Proposition}
\newtheorem{definition}[theorem]{Definition}
\newtheorem{corollary}[theorem]{Corollary}

\title{\Large A geometric approach to quantum control\\ in projective Hilbert spaces}
\author{ Davide Pastorello\\\normalsize Department of  Mathematics, University of Trento,\\
Istituto Nazionale di Fisica Nucleare (TIFPA)\\
 \normalsize via Sommarive 14, 38123 Povo (Trento), Italy.}

\date{}

\begin{document}

%\maketitle

%\twocolumn[\begin{@twocolumnfalse}

\maketitle
\begin{abstract}
\noindent A quantum theory in a finite-dimensional Hilbert space can be geometrically  formulated as a proper Hamiltonian theory as explained in \cite{AS,BH,DV2,D1}. From this point of view a quantum system can be described in a classical-like framework where quantum dynamics is represented by a Hamiltonian  flow in the phase space given by projective Hilbert space. This  paper is devoted to investigate how the notion of \emph{accessibility algebra} from classical control theory can be applied within geometric Hamiltonian formulation of Quanum Mechanics to study controllability of a quantum system. %in order to state the following conjecture: Under certain conditions, classical control theory provides a machinery which can be directly applied in quantum control within the geometric Hamiltonian picture. 
A new characterization of quantum controllability in terms of Killing vector fields w.r.t. Fubini-Study metric on projective space  is also discussed.

\end{abstract}

\vspace{0.5cm}


%\end{@twocolumnfalse}]

\vspace{0.5cm}	
\section{Introduction}
%A classical control system can be defined as a classical mechanical system described by the Hamiltonian function:
%\beq
%H=H_0+\sum_{i=1}^m H_i u_i
%\eeq

\noindent
In Quantum Control Theory a fundamental issue is understanding if a quantum system can be driven from an initial state to a selected finale state, i.e. if it is controllable (some notions of quantum controllability are sketched below). Quantum Control obviously requires taking into account aspects like entanglement, coherence, unitary evolution, dynamics of open systems that are quantum features without a counterpart in Classical Mechanics. However adopting a geometric point of view, a $n$-level quantum system can be described as a Hamiltonian system in a symplectic manifold in analogy to a classical system. Since Classical Control Theory is rich of tools within a geometric formulation we can conjecture that some of these  tools can be directly applied to study quantum controllability. 
\\
The main goal of this paper is investigate how quantum controllability can be characterized applying the machinery of classical control theory within geometric formulation of Quantum Mechanics. In particular theorem \ref{T9} shows that complete controllability of a quantum system can be tested checking classical \emph{accessibility rank condition}. Another aim of the work is investigate how geometric structure of Hilbert projective space can be used to study quantum control, in particular we show that controllability of a quantum system can be characterized in terms of Killing vector fields w.r.t. riemannian metric defined on quantum phase space.
\\
In section 3, we summarize geometric Hamiltonian formulation of finite-dimensional QM on projective space construced out from the Hilbert space of the considered quantum theory. In particular a general prescription to associate a phase space function (classical-like observable) to every quantum observable is stated. Within the symplectic structure of the projective space, classical-like observables define Hamiltonian vector fields. In the present paper we set up a description of quantum controllability in terms of Hamiltonian vector fields as in Classical Mechanics.

\noindent   
Let us start introducing some fundamental concepts to investigate if the dynamics of a quantum physical system can be controlled acting on external controls which are described by coupling terms in the system Hamiltonian operator. 
\\
The bilinear model of \emph{controllability problem} for a multilevel quantum system described in a $n$-dimensional Hilbert space $\sH$ is given by:
\beq\label{qc}
i\hbar\frac{d}{d t}|\psi(t)\rangle=\left[H_0+\sum_{i=1}^m H_i u_i(t)\right] |\psi(t)\rangle.
\eeq
(\ref{qc}) is the Schr\"odinger equation defined by  a total Hamiltonian made by an internal part $H_0$ and  coupling terms with control functions $u_i=u_i(t)$ for $i=1,...,m$. The ket $|\psi\rangle$ is a normalized wavefunction, so a curve $t\mapsto |\psi(t)\rangle$ on the unit sphere $\bS(\sH):=\{|\psi\rangle\in\sH|\langle\psi|\psi\rangle=1\}$ satisfying (\ref{qc}) represents the time evolution of a pure state.

\begin{definition}\label{defpsc}
The quantum system (\ref{qc}) is {\bf pure state controllable} if for every pair $|\psi_0\rangle$, $|\psi_1\rangle\in\bS(\sH)$ there exist control functions $u_1,...,u_m$ and a time $T>0$ such that 
the solution of (\ref{qc}), with initial condition $|\psi(0)\rangle=|\psi_0\rangle$, at time $T$ is $|\psi(T)\rangle=|\psi_1\rangle$.
\end{definition}
\noindent
Since state vectors differing by a phase factor are physically indistinguishable i.e. the set of the pure states is isomorphic to the space of rays in $\sH$, thus the above definition should consider the quotient of the sphere $\bS(\sH)$ w.r.t. the group $U(1)$, i.e. the projective space on $\sH$. However pure state controllability is equivalent to the controllability in terms of rays (sometimes called \emph{equivalent state controllability}), as proved in \cite{AD}.  
\\
According to standard formulation of QM,  solution at time $t$ of (\ref{qc}) with initial condition $|\psi_0\rangle$ is given by the unitary evolution (assuming the considered system is not interacting with enviroment):
\beq
|\psi(t)\rangle=U(t)|\psi_0\rangle,
\eeq
where $U(t)$ belongs to the unitary group $U(n)$ and it is given by the solution at time $t$ of:

\beq\label{oqc}
i\hbar\frac{d}{d t}U(t)=\left[H_0+\sum_{i=1}^m H_i u_i(t)\right]U(t),
\eeq
\\
with initial condition $U(0)=I_n$.
Another notion of controllability of physical interest is the so called \emph{operator controllability} or \emph{complete controllability} \cite{AD,DD,SS}.
\begin{definition}
The quantum system (\ref{qc}) is \textbf{\bf operator controllable} or \textbf{\bf completely controllable} if every unitary operation on the initial state can be performed by an appropritae choice of the controls $u_1,...,u_m$. In other words, for every $U_f\in U(n)$ there exist controls $u_1,...,u_m$ and $T>0$ such that the solution of (\ref{oqc}), with initial condition $U(0)=I_n$, satisfies  $U(T)=U_f$.

\end{definition}

\noindent
Let $\gB(\sH)$ be the linear Lie algebra of operators in the Hilbert space $\sH$ and $\{A_1,...,A_m\}$ be a finite set of operators in $\gB(\sH)$, the smallest subalgebra of $\gB(\sH)$ containing $A_1,...A_m$ is called \textbf{algebra generated by} $\{A_1,...A_m\}$.  A sufficient and necessary condition for operator controllability of a quantum system is given by the following result \cite{AD,SS}:

\begin{theorem}\label{qcontrol}
The system (\ref{qc}) is operator controllable if and only if the Lie algebra $\mL$ generated by $\{-iH_0,-iH_1,...,-iH_m\}$ is $\gu(n)$, the Lie algebra of the unitary group $U(n)$.
\end{theorem}
\noindent
The Lie algebra $\mL$ is tipically called \textbf{dynamical algebra} of the system (\ref{qc}).
The goal of the present paper is to show that a quantum system is complete controllable if and only if it satisfies the \emph{classical accessibility rank condition} (theorem \ref{ccontrol}) within the classical-like geometric Hamiltonian formulation of QM that is introduced in section \ref{geomQM}.

\vspace{0.0cm}

\section{Controllability and local accessibility of a classical system}

In this section let us recall some fundamental concepts about classical control theory within Hamiltonian formalism, we consider a classical system coupled with a set of external controls. Under certain conditions a suitable choice of these controls can allow to drive the system in the phase space along preferred trajectories and such a control can be a local or a global property. 
The general \emph{controllability problem} for a classical non-linear system in the affine form can be stated as follows:
\beq\label{cc}
\dot x (t)=X_0(x(t))+\sum_{i=1}^m X_i(x(t)) u_i(t), 
\eeq
where $[0,+\infty)\ni t\mapsto x(t)$ is a curve in the $n$-dimensional smooth symplectic manifold $\cM$ (the phase space, $n$ is the total number of degrees of freedom of the system) describing time evolution of the classical system, $X_0,X_1,...,X_m$ are smooth vector fields on $\cM$ and $u_1,...,u_m$ are real valued functions called \emph{control functions} or simply \emph{controls}. In (\ref{cc}) there are the Hamilton equations of motion defined by the Hamiltonian function 
$$H_0+\sum_{i=1}^m H_i u_i(t).$$
More precisely: Let $\omega$ be the symplectic form defined on $\cM$, $X_i$ is the unique vector field satisfying $\omega(X_i,Y)=dH_i(Y)$ for every vector field $Y$ on $\cM$, i.e. $X_i$ is the Hamiltonian vector field associated to the smooth function $H_i:\cM\rightarrow \bR$ (for any $i=1,...,m$).
\\
Control on the system is realized by the set of  controls $\{u_1,...,u_m\}$ coupled with the system via time-independent Hamiltonians $H_i$. The function $H_0$ represents the internal energy of the system and its Hamiltonian vector field $X_0$ is sometimes called \emph{drift vector field}.   
Generally speaking the goal of a control problem is finding a set of controls $\{u_1,...,u_m\}$ in order to drive the system from an initial state $x(0)\in\cM$ to a final state $x(T)\in\cM$, in a finite time $T>0$. In the following we refer to (\ref{cc}) as the \emph{control system} or simply as the \emph{system}.
\begin{definition}
The system (\ref{cc}) is said to be {\bf controllable} if for any two points $x_1,x_2\in\cM$ there exists a finite time $T$ and controls $u_1,...,u_m$ such that the solution $x=x(t)$ of (\ref{cc}) with initial condition $x(0)=x_1$ satisfies $x(T)=x_2$.
\end{definition}

\noindent
Hence if a system is controllable then any physically permitted evolution can be achieved by a controller which acts on controls.\\
Let $\cR^V(x_0,T)$ be the subset of $\cM$ made by points that are reachble from $x_0$ at time $T$ following trajectories contained in the neighborhood $V$ of $x_0$, i.e. being $x=x(t)$ the solution of (\ref{cc}) with initial condition $x(0)=x_0$ we can define:  

\beq
\cR^V(x_0,T):=\{\hat x\in\cM|\exists u_1,...,u_m\,\,\mbox{s.t.}\,\, x(t)\in V,0\leq t\leq T, x(T)=\hat x\}.
\eeq
\\
Defining the reachable set $\cR^V_T(x_0):=\bigcup_{\tau\leq T} \cR^V(x_0,\tau)$, we can state the following definition.
\begin{definition}
The control system (\ref{cc}) is said to be {\bf locally accessible from $x_0\in\cM$} if $\cR_T^V(x_0)$ contains a non-empty open set of $\cM$ for all  $V$ and $T>0$. If such property holds for all $x_0\in\cM$ then the system is said to be {\bf locally accessible}.
\end{definition}

\noindent
Thus system is controllable if $\bigcup_{T>0} \cR^\cM_T(x_0)=\cM$ for all $x_0$.
A useful tool to study controllability and local accessibility of a system is given by a particular subalgebra of the Lie algebra $V^{\infty}(\cM)$ of smooth vector fields on $\cM$.

\begin{definition}
The {\bf accessibility algebra} $\cC$ of the system (\ref{cc}) is the subalgebra of $V^{\infty}(\cM)$ generated by $X_0,X_1,...,X_m$.
\end{definition}
\noindent
So $\cC$ is the smallest subalgebra of $V^{\infty}(\cM)$ containing vector fields $X_0,X_1,...,X_m$. Let us introduce a linear subspace of the tangent space $T_x\cM$ in $x\in\cM$ associated to $\cC$ which defines the so called \emph{accessibility distribution}:
\beq
\cC(x):=span\{X(x)|X\in\cC  \}
\eeq
The following result is the celebrated \emph{accessibility rank condition} \cite{book}:
\begin{theorem}\label{ccontrol}
Consider a classical control system (\ref{cc}). If $\dim \cC(x)=n$ for every $x\in\cM$ then the system is locally accessible.
\end{theorem}

\noindent
In other words the condition of local accessibility is equivalent to the following fact: The vector fields belonging to accessibility algebra of the system span the tangent space in every point of the manifold. 
%In the classical theory local accessibility implies controllability only in case of null drift vector field $X_0=0$.
\\
Below we apply classical accessibility rank condition to characterize the controllability of a quantum system which is described as a classical-like system in the geometric Hamiltonian framework, where $\cM$ is given by the projective space of the Hilbert space of a $n$-level system.

\section{Geometric Hamiltonian formulation of Quantum Mechanics}\label{geomQM}

Finite dimensional Quantum Mechanics can be formulated as a classical-like Hamilonian theory where the phase space is given by the projective space  $\sP(\sH)=\bS(\sH)/U(1)$ on the Hilbert space $\sH$ of considered quantum theory (e.g. \cite{AS}). $\sP(\sH)$ has a structure of ($2n-2$)-dimensional real manifold  which can be equipped with a \emph{K\"ahler structure}. In particular we introduce the symplecic form below.
% In particular the interesting symplectic form is given by the \emph{Konstant-Kirillov form} considering $\sP(\sH)$ as a rank-1 orbit of the unitary group $U(n)$. 
Points of $\sP(\sH)$ can be represented by rank-1 orthogonal projectors on $\sH$, i.e. pure states, in fact it is well-known that there exists a homeomorphism $\sP(\sH)\rightarrow \gD_p(\sH)$ where $\gD_p(\sH)$ denotes the set of pure states endowed with the standard operator topology\footnote{Let $\gB(\sH)$ be the $C^*$-algebra of linear operators on the finite-dimensional Hilbert space $\sH$, $\gD_p(\sH)$ is defined as the set of extremal points of $\gD(H)\subset\gB(\sH)$ that is the convex set of positive operators with unit trace, i.e. density matrices (in finite dimension every operator is obviously bounded and trace-class). Standard operator topology is the topology equivalently induced by \emph{$C^*$-norm} or \emph{trace norm} or \emph{Hilbert-Schmidt norm} defined on $\gB(\sH)$.} and $\sP(\sH)$ is equipped with the quotient topology\footnote{In the quotient topology, $\sP(\sH)$ is connected and Hausdorff.}.
\\
In order to define explicitely the symplectic form on $\sP(\sH)$ let us introduce the following characterization of the tangent space $T_p\sP(\sH)$ which is made possible by the transitive action of $U(n)$ on $\sP(\sH)$.  In the following $i\gu(n)$ denotes the real vector space of self-adjoint operators in $\sH$.
\begin{proposition}\label{tang}
The tangent vectors  $v$ at $p \in \sP(\sH)$ are all the linear operators on $\sH$ of the form:
$v = -i[A_v, p]$, for some  $A_v\in i\gu(n)$.
Consequently,  $A_1,A_2 \in i\gu(n)$ 
define the same vector in $T_p \sP(\sH)$  iff $[A_1-A_2, p]=0$. 
\end{proposition}
\begin{proof}
$\sP(\sH)$ is diffeomorphic to the quotient $U(n)/\mathcal G_p$ where $\cG_p\subset U(n)$ is the isotropy subgroup of $p\in \sP(\sH)$.
Let us consider the transitive smooth action of the compact Lie group $U(n)$ on the projective space $\sP(\sH)$ (identifying points of $\sP(\sH)$ with rank-1 orthogonal projectors on $\sH$):
\beq
U(n)\times \sP(\sH) \ni (U,p)\mapsto Up\,U^{-1}\in\sP(\sH).
\eeq
The projection defined as:
$$\Pi_p:U(n)\ni U\mapsto U p\, U^{-1} \in\sP(\sH)$$
 is a submersion \cite{Wa} and thus $d\Pi_p|_{U=I}:\gu(n)\rightarrow T_p\sP(\sH)$ is surjective. Since $d\Pi_p(B)|_{U=I}=[B,p]$ for every $B\in\gu(n)$, i.e. for every anti-self adjoint operator on $\sH$, the claim is proved.
\\
\end{proof}

\noindent
As a consequence of the above result we have this isomorphism of vector spaces:
\beq \label{ts}
T_p\sP(\sH)\simeq\frac{i\gu(n)}{\sim_p}
\eeq
where the equivalence relation $\sim_p$ is defined by:
\beq
 A_1\sim_p A_2\Leftrightarrow\, [A_1-A_2,p]=0.
\eeq
A symplectic form on $\sP(\sH)$ can be defined for any $\kappa>0$ as:
\beq\label{sp}
\omega_p(u,v):=-i\kappa tr(p[A_v,A_u])\quad\quad u,v\in T_p\sP(\sH)
\eeq
and also a Riemannian metric is defined:
\beq\label{FS}
g_p(u,v)=-\kappa tr(p([A_u,p][A_v,p]+[A_v,p][A_u,p]))\qquad \kappa>0,
\eeq
called {\bf Fubini-Study metric}. Adopting the notation $-i[A,p]=dp$ and applying the polarization identity, (\ref{FS}) can be written in the celebrated form $ds^2=g_p(dp,dp)=2\kappa tr(p(dp)^2)$. These two definitions are well-posed and gives tensors even if the operator $A_v$ associated to $v$ is not unique because the right-hand sides of (\ref{sp}) and (\ref{FS}) are fixed if adding to $A_v$ or $A_u$ operators commuting with p. 
 The symplectic form  is comapatible with Fubini-Study metric w.r.t. to a complex form \cite{AS,BH,DV2,D1}, so $\sP(\sH)$ has a structure of K\"ahler manifold. In particular the symplectic structure can be exploited to construsct a proper Hamiltonian mechanics: For every smooth function $f:\sP(\sH)\rightarrow \bR$ the associated Hamiltonian vector field $X_f$ is the unique vector field satisfying $\omega_p(X_f,\,\cdot\,)=df_p$ and the notion of \emph{Poisson bracket} of a pair of smooth functions $f,g:\sP(\sH)\rightarrow\bR$ can be defined as $\{f,g\}:=\omega(X_f,X_g)$. The formula $[X_f,X_g]=X_{\{f,g\}}$ holds, where the commutator $[\,\,\,,\,\,]$ is the Lie bracket of vector fields.
\\
The main idea of geometric Hamiltonian formulation of QM is associating any quantum observable $A\in i\gu(n)$ to a real scalar function $f_A$ on $\sP(\sH)$ (i.e. a classical-like observable given by a phase space function) in order to obtain a classical-like description of a quantum system on the projective space, in particular representing quantum dynamics via a Hamiltonian vector field w.r.t. the symplectic form (\ref{sp}).
\\
 Imposing several physical requirements \cite{DV2,D1} a general prescription to set up a meaningful classical-like Hamiltonian formulation of a quantum theory is given by the so-called \emph{inverse quantization maps} $\cO$ and $\cS$:
\begin{equation}\label{O}
\cO:i\gu(n)\ni A\mapsto f_A,
\end{equation}
with 
\begin{equation}\label{o}
f_A(p)=\kappa tr(Ap)+\frac{1-\kappa}{n}tr(A)\qquad \kappa>0,
\end{equation}
\noindent
for observables. And about states:
\begin{equation}\label{S}
\cS:\gD(\sH)\ni \sigma\mapsto \rho_\sigma
\end{equation}
with 
\begin{equation}\label{s}
\rho_\sigma(p)=\kappa' tr(\sigma p)+\frac{\kappa-(n+1)}{\kappa}\qquad \kappa>0
\end{equation}

\vspace{0.3cm}
\noindent
where $\kappa'=\frac{n(n+1)}{\kappa}$ and $\gD(\sH)$ denotes the set of density matrices on $\sH$.  Indeed we have a one-parameter family of prescriptions $\{\cO_\kappa,\cS_\kappa\}_{\kappa>0}$ to set up a proper Hamiltonian theory. For $\kappa=1$, the observable $A$ is represented by the standard expectation value function $f_A(p)=tr(Ap)$ as in the Ashtekar-Schilling picture \cite{AS}. While the choice $\kappa=n+1$ is convinient if one needs a simple form of Liouville densities representing quantum states like within geomeric description of entangled states presented in \cite{D}.\\
Using the maps $\cO$ and $\cS$ to obtain classical-like observables and states we can compute quantum expectation values as classical expectation values:
\begin{equation}
\langle A\rangle_\sigma=tr(A\sigma)=\int_{\sP(\sH_n)} f_A \rho_\sigma d\nu
\end{equation}
where $\nu$ is the suitably normalized Liouville measure induced by the symplectic form $\omega$ (see \cite{DV1,DV2} for a complete description of $\nu$). For any $A\in i\gu(n)$, the Hamiltonian vector field associated to $f_A=\cO(A)$ within the symplectic structure induced by (\ref{sp}) is given by \cite{AS,BH,DV2}:
\beq\label{field}
X_{f_A}(p)=-i[A,p]\quad\forall p\in\sP(\sH).
\eeq
A very remarkable result is stated below \cite{AS,BH}:
\begin{theorem}\label{killing}
A vector field on $\sP(\sH)$ is the Hamiltonian vector field of a classical-like observable if and only if it is a $g$-Killing vector field.
\end{theorem}

\noindent
Thus one can note the following fact: If $A$ is the Hamiltonian operator $H$ then the Hamilton dynamics given by the field $X_{f_H}$ is equivalent to the Schr\"odinger dynamics given by $H$, i.e. a curve $t\mapsto p(t)\in\sP(\sH)$ satisfies the Schr\"odinger equation ($\hbar =1$):
\beq
\dot p(t)=-i[H,p(t)],
\eeq  
if and only if it satisfies the Hamilton equation:
\beq 
\dot p(t)=X_{f_H}(p(t)).
\eeq

\noindent
In particular, if $\dim\sH>2$ then we have a $C^*$-isomorphism induced by $\cO$ (for any $\kappa>0$)between the $C^*$-algebra of linear operators in $\sH$ and the $C^*$-algebra of certain square $\nu$-integrable functions (so-called \emph{frame functions} \cite{DV1}) on projective space \cite{DV2,D1}. In this case the prescription introduced in (\ref{O})-(\ref{s}) is the \emph{unique} procedure to set up a Hamiltonian formulation of a quantum theory on projective space. Thus we can construct a concrete observable algebra in terms of phase space functions (abandoning self-adjoint operators) obtaining a self-consistent classical-like Hamiltonian formulation of finite-dimensional Quantum Mechanics with a complete characterization of mixed states as normalized probability densities. In terms of such \emph{Liouville densities}, a description of entangled states is given in \cite{D} with focus on quantum information theory.

\vspace{0.5cm}

\section{Complete quantum controllability in terms of classical-like local accessibility}

Since Schr\"odinger dynamics is equivalent to Hamilton dynamics within the setting described in the previous section, so we can state a quantum control problem in terms of Hamiltonian vector fields on $\sP(\sH)$. In this picture,  system (\ref{qc}) (for $\hbar=1$) is equivalent to the nonlinear system:
\beq\label{cqc}
\dot p(t)=X_0(p(t))+\sum_{i=1}^m X_i(p(t)) u_i(t),
\eeq 
where $p(t)=|\psi(t)\rangle\langle\psi(t)|$ and $X_0,...,X_m$ are the Hamiltonian vector fields on $\sP(\sH)$ defined by the Hamiltonian functions $f_{H_0}=\cO(H_0),...,f_{H_m}=\cO(H_m)$:
\beq
f_{H_i}(p)=\kappa tr(H_i p)+\frac{1-\kappa}{n}tr(H_i)\quad\quad \kappa>0.
\eeq
%Note that two selfadjoint operators $A_1$ and $A_2$ defines the same vector field if and only if $[A_1-A_2,p]=0$ for every rank-1 orthogonal projector $p$, i.e. $A_1-A_2$ belongs to the center of $i\gu(n)$.  So the Lie algebra $\cV(\sH)$ of Hamiltonian vector fields on $\sP(\sH)$ is isomorphic to the quotient of $i\gu(n)$ w.r.t. the equivalence of elements whose difference is in the center.
%Let us consider the Lie algebra $\cV(\sH)$ of  Hamiltonian vector fields on $\sP(\sH)$. We have a homomorphism $i\gu(n)\ni A\mapsto X_A\in\cV(\sH)$ of Lie algebras, in fact there is a bijective correspondence: self-adjoint operators $A$ and $B$ define the same Hamiltonian vector field if and only if $[A-B,p]=0$ for every $p\in\sP(\sH)$ and this reads $A=B$, then the Lie-algebraic structures are compatible in this way:
We call (\ref{cqc}) \emph{classical-like system} equivalent to the quantum system (\ref{qc}) because its described in terms of Hamiltonian vector fields on a quantum phase space within the classical-fashioned formalism introduced in  the previous section.\\
Let us point out that the Poisson bracket of $f_A=\cO(A)$ and $f_B=\cO(B)$, with $A,B\in i\gu(n)$, is $\{f_A,f_B\}=f_{-i[A,B]}=\cO(-i[A,B])$, so if $X_A$, $X_B$ and $X_{-i[A,B]}$ are respectively the Hamiltonian vector fields associated to $f_A$, $f_B$ and $f_{-i[A,B]}$ then we have the following remarkable identity chain: 
\beq\label{fields}
[X_A,X_B](p)=X_{-i[A,B]}(p)=-i[-i[A,B],p]=[[-iA,-iB],p] \qquad \forall p\in \sP(\sH),
\eeq
where the commutator in the first member is the Lie bracket of vector fields. Hence we have a Lie algebraic homomorphism between the algebra of Hamiltonian vector fields corresponding to quantum observables (that are the Killing vector fields w.r.t. Fubini-Study metric) and $\gu(n)$, the algebra of anti-selfadjoint operators. More precisely $\gu(n)\ni T\mapsto X$ with $X(p)=[T,p]$ is an epimorphism of Lie algebras.
\\
Let $\cV(\sH)$ be the algebra of smooth vector fields on $\sP(\sH)$, for the system (\ref{cqc}) the accessibility algebra $\cC$ is defined as the subalgebra of $\cV(\sH)$  generated by the fields $\{X_0,...,X_m\}$. We can define accessibility distribution on projective space:
\beq\label{C(p)}
\cC(p)=span\{X(p)|X\in\cC\}\subset T_p\sP(\sH).
\eeq
In order to show that the classical accessibility condition $\cC(p)=T_p\sP(\sH)$ $\forall p\in\sP(\sH)$ for the system (\ref{cqc})  implies operator controllability for the system (\ref{qc}), let us present the following technical result \cite{DD, book}:
\begin{proposition}\label{10}
Let $\mathfrak A$ be a Lie algebra with Lie brackets $[\,\,\, ,\,\,\,]$ and $\{a_1,...,a_r\}\subset \mathfrak A$ be a finite subset. Let $\mathfrak B$ be the algebra generated by $\{a_1,...,a_r\}$ i.e. the smallest subalgebra of $\mathfrak A$ that contains $a_1,...,a_r$. Then elements of $\gB$ turns out to be all linear combinations of repeated Lie brackets:
\beq\label{rlb}
[A_k,[A_{k-1},[\cdots[A_2,A_1]\cdots]]]\qquad k\in\bN,
\eeq
where $A_\alpha\in\{a_1,...,a_r\}$ for any $\alpha\in\{1,...,k\}$.
\end{proposition}

\noindent
The statement of following theorem shows that if a quantum system is described within geometric Hamiltonian formulation so that associated accessibility algebra can be defined then accessibility rank condition (stated like in classical control theory) is satisfied if and only if the considered quantum system is completely controllable. Let us recall that referring to (\ref{qc}) the dynamical Lie algebra $\mL$ of a quantum system is defined as the operator algebra generated by $\{-iH_0,...,-iH_m\}$. Let us stress the control problem can be translated in geometric Hamiltonian formulation considering the Hamitonian vector fields $X_0,...,X_m$ on the projective space generated by the functions $f_{H_0}=\cO(H_0),...,f_{H_m}=\cO(H_m)$ obtaining Hamilton equation (\ref{cqc}).

\begin{theorem}\label{T9}
Consider a quantum system described in the $n$-dimensional Hilbert space $\sH$ whose dynamics is governed by (\ref{qc}). Let $\cC$ be the accessibility algebra of the system within geometric Hamiltonian description on projective space $\sP(\sH)$. The system is completely controllable if and only if the accessibility rank condition is satisfied, i.e.:
\beq\label{rank}
T_p\sP(\sH)=span\{X(p)|X\in\cC\}\qquad\forall p\in\sP(\sH).
\eeq
\end{theorem}
\begin{proof}
\def\mX{\mathcal X}
\def\mH{\mathcal H}
We prove that condition (\ref{rank}) is satisfied if and only if $\mL=\gu(n)$, where $\mL$ is the dynamical Lie algebra of the quantum system and $\gu(n)$ is the Lie algebra of the unitary group $U(n)$. First of all let us prove that a vector field $X$ on $\sP(\sH)$ belongs to $\cC$ if and only if there is an operator $T\in\mL$ such that $X(p)=[T,p]$ for every $p\in\sP(\sH)$.
\\
Let $X_0,X_1,...,X_m$ be the Hamiltonian vector fields on $\sP(\sH)$ associated to quantum operators $H_0,H_1,...,H_m$, then $\cC$ is generated by  $\{X_0,X_1,...,X_m\}$ by definition. By proposition \ref{10} any element $X$ of $\cC$ can be written as a linear combination $X=\sum_k a_k \mathfrak X_k$ of fields given by repeated Lie brackets:
\beq\label{vf}
\mathfrak X_k =[\mX_k,[\mX_{k-1},[\cdots[\mX_2,\mX_1]\cdots]]]\qquad k\in\bN,
\eeq
where $\mX_i\in\{X_0,...,X_m\}$ for any $i=1,...,k$. Using (\ref{field}) we can write the action of the vector field $\mathfrak X_k$ on $p\in\sP(\sH)$ as:
$$\mathfrak X_k(p)= [\mX_k,[\mX_{k-1},[\cdots[\mX_2,\mX_1]\cdots]]]  (p)=$$
$$\,\,\,\, \,\,\,\qquad\qquad\qquad = [[-i\mH_k,[-i\mH_{k-1},[\cdots[-i\mH_2,-i\mH_1]\cdots],p]]], $$
\\
with $\mH_i\in\{H_0,...,H_m\}$ for any $i=1,...,k$ according to (\ref{fields}). Hence $X(p)=[T,p]$ for every $p\in\sP(\sH)$, where $T$ is some anti-selfadjoint operator given by a linear combination of terms of the following form:
\beq\label{27}
\mathfrak T_k=[-i\mH_k,[-i\mH_{k-1},[\cdots[-i\mH_2,-i\mH_1]\cdots]
\eeq
thus $T$ belongs to the subalgebra $\cL$ by proposition \ref{10}. The converse is true, if $T\in \cL$ then it can be written as a linear combination of terms (\ref{27}) and the associated vector field $X(p)=[T,p]$ is a linear combination of repeated Lie brackets (\ref{vf}), i.e. it belongs to $\cC$.

\noindent
Using this result, the accessibility distribution $\cC(p):=span\{X(p)|X\in\cC\}$ can be written as:
\vspace{0.0cm}
\beq\label{access}
\cC(p)=span\{[T,p]|T\in\cL\}\qquad\forall p\in\sP(\sH).
\eeq
If we require the \emph{accessibility condition} $\cC(p)=T_p\sP(\sH)$ $\forall p\in\sP(\sH)$, then $\cC(p)$ is isomorphic to $i\gu(n)/\sim_p$ for every $p\in\sP(\sH)$ as provided by (\ref{ts}), thus denoting by $[iT]_p$ the equivalence class of the selfadjoint operator $iT$ in  $i\gu(n)/\sim_p$, we have:
\beq
\frac{i\gu(n)}{\sim_p}=span\{[iT]_p|T\in \cL\} \quad \forall p\in\sP(\sH).
\eeq
Hence $\cL=\gu(n)$. Then classical local accessibility condition within the formulation given by (\ref{cqc}) implies complete controllability of the  quantum system (\ref{qc}). Moreover, if the Lie algebra generated by $\{-iH_0,...,-iH_1\}$ is $\gu(n)$ then $C(p)$ span the tangent space in $p$ everywhere on $\sP(\sH)$ as a cosequence of (\ref{access}). Therefore we can conclude that the accessibility rank condition is a necessary and sufficient condition for the complete controllability of the quantum system (\ref{qc}). 
%\begin{theorem}\label{ccc}
%A quantum system is completely controllable if and only if the correspondent classical-like system satisfies the accessibility rank condition.
%\end{theorem}
\end{proof}
\noindent
Let us stress the following fact: Even if condition (\ref{access}) is essentially a condition on the tangent space of $\sP(\sH)$ (whose points are identified with pure states) it is not only connected  with pure state controllability but with stronger property of operator controllablity, this non-trivial result is a consequence of the correspondence between tangent vectors in $T_p\sP(\sH)$ and self-adjoint operators on $\sH$, as shown by proposition  \ref{tang}. In other words operator controllability can be completely studied by means of a necessary and sufficient condition on the tangent space by virtue of transitive action of the unitary group on $\sP(\sH)$. On the other hand we focus on the notion of pure state controllability within our framework in the next section showing its weakness w.r.t. complete controllability in terms of accessibility algebra as an algebra of Killing fields.
\\
Theorem \ref{T9} allows to characterize the dynamical Lie algebra $\cL$ associated to a quantum control system in terms of the accessibility algebra defined within geometric Hamiltonian formulation.
Applying the well-known result of theorem \ref{killing}, we can state an additional characterization of quantum controllability in terms of accessibility algebra:  There is complete controllability if and only if $\cC$ is the Lie subalgebra of $\cV(\sH)$ made by the Hamiltonian vector fields on $\sP(\sH)$ corresponding to quantum observables that are the $g$-Killing vector fields ($g$ is Fubini-Study metric), we denote such algebra as $\mathfrak{Kill}(\sP(\sH))$. Then let us state:
\begin{corollary}\label{t9}
A quantum system is completely controllable  if and only if the accessibility algebra of the correspondent classical-like system is $\mathfrak{Kill}(\sP(\sH))$.
\end{corollary}
\noindent
%While a quantum system turns out to be only pure state controllable if and only if its accessibility algebra is a subalgebra of $\mathfrak{Kill}(\sP(\sH))$ that is isomorphic to the Lie algebra of a group whose action is transitive on $\sP(\sH)$, as explained in the next section in purely geometric terms.

\vspace{0.0cm}

\section{Pure state controllability in terms of classical-like accessibility algebra}
\noindent
In conclusion, we focus on the notion of  \emph{pure state controllability} (definition \ref{defpsc}), that is a property weaker than complete controllability. By definition, pure state controllability of a quantum system is equivalent to the following fact: The Lie group $e^{\cL}$ (where $\cL$ is the dynamical algebra) has a transitive action on the unit sphere $\bS(\sH)$. As proved in \cite{AD}, such a condition is equivalent to one of the following facts: $\cL$ is isomorphic to $\mathfrak {sp}(\frac{n}{2})$ or to $\mathfrak{su}(n)$, for $n$ even, or to $\mathfrak{su}(n)$, for $n$ odd. Where  $\mathfrak {sp}(\frac{n}{2})$ is the Lie algebra of the symplectic group and $\mathfrak{su}(n)$ is the Lie algebra of $SU(n)$. However, a more practical pure state controllability criterion is also introduced in \cite{AD}: Consider a pure state $P\in\gD_p(\sH)$ (which can be identified as a point $p\in\sP(\sH)$), let $\mathfrak C_P$ the centralizer of $iP$ in $\gu(n)$ (i.e. the subalgebra of $\gu(n)$ made by the operators $T\in\gu(n)$ such that $[T,iP]=0$). A quantum system with dynamical algebra $\cL$ is pure state controllable if and only if: 
\beq\label{standard}
\dim\cL-\dim(\cL\cap \mathfrak C_P)=2n-2,
\eeq
independently of the choice of $P$ because the transitive action of unitary group on the set of pure states. This test is really convenient when performing calculations in the Hilbertian basis where $P=$diag$(1,0,\dots,0)$. Let us show there is an analogous condition in terms of accessibility algebra that can be proved exploiting the geometric structure of $\sP(\sH)$ and the characterization of quantum observables by means of Hamiltonian vector fields. 
\\
The Lie algebra of Killing vector fields on a compact (pseudo-)Riemannian manifold corresponds to the Lie algebra of the isometry group of the manifold. In the case of $\sP(\sH)$ the isomtery group is the projective unitary group $PU(n)$ then we have $\dim(\mathfrak{Kill}(\sP(\sH)))=n^2-1$. One can use another simple argument to calculate the dimension of Lie algebra of Killing vector fields on $\sP(\sH)$ as a Riemannian manifold: We can consider the Lie algebraic homomorphism $\phi:\gu(n)\ni T\mapsto X_T\in \cV(\sH)$ where $X_T(p):=[T,p]$ for any $p\in\sP(\sH)$, let $f_{iT}$ be the classical-like observable associated to $iT$ for any $T\in\gu(n)$ (then $X_T$ is the Hamiltonian vector field of $f_{iT}$), since $\sP(\sH)$ is connected the map $f_{iT}\mapsto X_T$ is a Lie algebraic homomorphism whose kernel is given by the constant functions. Hence the kernel of  $\phi$ is the 1-dimensional ideal $\mathfrak I$ of $\gu(n)$ given by the multiples of $iI_n$. Thus the dimension of the range of $\phi$ (that corresponds to $\mathfrak{Kill}(\sP(\sH)$ in view of theorem \ref{killing}) is $n^2-1$.\\
 The following theorem establishes a necessary and sufficient condition for pure state controllability in terms of the accessibility algebra. It turns out to be the analogous of the standard one (\ref{standard}) where the dynamical Lie algebra $\mL$ is replaced by the accessibility algebra of the classical-like formulation and the centralizer $\mathfrak C_P$ is replaced by a subalgebra of vanishing Killing vector fields.
\begin{theorem}\label{dim}
A quantum system described within the geometric Hamiltonian picture is pure state controllable if and only if
\beq\label{psc}
\dim(\cC)-\dim(\cC\cap \mathfrak A_p)=2n-2
\eeq
for a fixed $p\in\sP(\sH)$. Where $\cC$ is the classical accessibility algebra and $\mathfrak A_p$ is the subalgebra of $\mathfrak{Kill}(\sP(\sH))$ of vector fields that satisfies $X(p)=0$.
\end{theorem}
\begin{proof}
We have seen that $\cC$ is a subalgebra of the Lie algebra of Killing vector fields in $\sP(\sH)$ then it is isomorphic to the Lie algebra of a one-parameter subgroup $\cG$ of the isometry group. The subalgebra $\cC\cap \mathfrak A_p$ of vanishing Killing vector fields in $p$ is isomorphic to the Lie algebra of a subgroup of the isotropy group of $p\in\sP(\sH)$ which belongs to $\cG$. Let us denote such isotropy subgroup with $\tilde\cG_p$. Since $\cG/\tilde\cG_p$ is diffeomorphic to the orbit $\cG.p$ of $p$ we have $\dim(\cC)-\dim(\cC\cap \mathfrak A_p)$ coincides with the dimension of $\cG.p$ that is an immersed submanifold of $\sP(\sH)$.
\\
If $\dim(\cC)-\dim(\cC\cap \mathfrak A_p)=\dim\sP(\sH)=2n-2$ then the orbit $\cG.p$ is $\sP(\sH)$ itself, hence  the action of $\cG$ is transitive on $\sP(\sH)$ and $p$ can be map to any point $p'\in\sP(\sH)$ under the action of a transformation $g\in\cG$, i.e. $p'=g(p)=U pU^*$ for some unitary operaor $U\in U(n)$. This fact is equivalent to pure state controllability, i.e. for any initial state and any final state one can find a time evolution from the first one to the second one. 
\\
The converse is true, if the considered quantum system is pure state controllable then there exists a one-parameter group of isometries $\cG$ with transistive action on $\sP(\sH)$ whose Lie algebra corresponds to the accessibility algebra $\cC$; Since $\sP(\sH)$ is a homogeneous space for $\cG$, the relation (\ref{psc}) for Lie algebras holds.   
\end{proof}

\vspace{0.5cm}

\noindent
As a remark note that theorem \ref{dim} represents an alternative way to prove the validity of the known pure state controllability criterion (\ref{standard}) (Theorem 7 in \cite{AD}), let us prove below that (\ref{standard}) is equivalent to (\ref{psc}).

\noindent
Consider the dynamical Lie algebra $\cL$ of a quantum system, applying theorem \ref{T9}, we have that $T\in\cL$ if and only if the vector field given by $X(p)=[T,p]$ belongs to the accessibility algebra $\cC$ (within the geometric Hamiltonian picture). In particular, if $T\in\cL$ is also an element of the centralizer of the operator $P\in\gD_p(\sH)$ identified by the point $p\in\sP(\sH)$ we have that associated Hamiltonian field $X\in\cC$ gives the null tangent vector in $p\in\sP(\sH)$. Viceversa is obviously true. For any Lie subalgebra $\mathfrak C$ of $\gu(n)$, the subalgebra $\phi(\mathfrak C)$ of $\mathfrak{Kill}(\sP(\sH))$ has dimension $\dim \phi(\mathfrak C)=(\dim\mathfrak C)-1$ if $\mathfrak I \subset \mathfrak C$ or $\dim\phi(\mathfrak C)=\dim\mathfrak C$ if $\mathfrak I \not\subset \mathfrak C$, where $\mathfrak I$ is the ideal formed by the multiples of $i I_n$. Note that if $\mathfrak I$ is a subalgebra of $\mL$ then $\mathfrak I$ is a subalgebra of $\mL\cap\mathfrak C_P$ because $\mathfrak I$ is an ideal of the centralizer.

\noindent
 Since $\cC=\phi(\mL)$ and $\cC\cap \mathfrak A_p=\phi(\mL\cap \mathfrak C_P)$ we have $\dim(\cC)-\dim(\cC\cap \mathfrak A_p)=2n-2$ if and only if $\dim\cL-\dim(\cL\cap \mathfrak C_P)=2n-2$ that is the standard condition for pure state controllability.

\noindent

\vspace{0.0cm}
\section{Conclusion and perspectives}
We can give the following interpretation of the result presented in section 4: since operator controllability of a quantum system can be expressed in terms of classical local accessibility then we conjecture that classical control theory can be directly applied to quantum control in several contexts  within  geometric Hamiltonian picture. In particular we have shown that accessibility algebra (as defined in classical control theory) plays a central r\^ole in geometric classical-like framework to characterize complete controllability and pure state controllability that are the main notions of controllability in quantum control theory.    
\\
In many concrete situations the quantum system to be controlled is an open quantum system and not an isolated one, then dynamics cannot be simply described by a group of unitary operators, in this case the study of controllability is more difficult and represents an extremely interesting direction of scientific investigation. In this regard a way to adopt a classical-like point of view (as suggested in the present work) is represented by \cite{D} where states of composite quantum systems and quantum entanglement are described in terms of Liouville densities.
\\
From physical point of view a remarkable open issue is the determination of a concrete class of control problems for which the presented geometric approach is definitely more convenient than the known methods of quantum control. An example of promising direction of investigation in this sense could be indirect controllability: A target system $S$ is coupled with an auxiliary system $A$ on which control is performed. Such a composite system can be completely described in Hamiltonian geometric terms \cite{D} where notion of partial trace is replaced by \emph{partial integral} w.r.t. a suitable Liouville form and a measure of quantum correlation in $S+A$ can be directly computed as a $L^2$-distance between Liouville densities. Hence this work provides a machinery which could lead to a new approach to indirect controllability.
\\
Several tools of classical optimal control theory, like Pontryagin minimum principle and variational methods, have been applied to optimal quantum control \cite{DP}; working in a geometric classical-like framework could be a fruitful strategy to adapt other classical tools to optimal control theory.

\vspace{0.5cm}

\section*{Acknowledgements}
I am grateful to D. D'Alessandro and V. Moretti for useful discussions and suggestions.


\vspace{0.5cm}



\begin{thebibliography}{        }
%\bibliographystyle{plainnat}

\bibitem[1] {AD} F. Albertini and D. D'Alessandro. {\em Notions of controllability for bilinear multilevel quantum systems}, IEEE Trans.Autom.Control, 48, pp.1399-1403 (2003).

\bibitem[2]{AS} A. Ashtekar and T.A. Schilling. {\em Geometry of quantum mechanics}, AIP Conference Proceedings, 342, 471-478 (1995).

\bibitem[3]{BH} D.C. Brody, L.P. Hughston. \emph{Geometric quantum mechanics}, Journal of Geometry and Physics, 38, 19-53 (2001).

\bibitem[4]{DD} D. D'Alessandro. {\em Introduction to Quantum Control and Dynamics}, Chapman \& Hall/CRC, 1st edn. (2007).


\bibitem[5]{DP} D.Dong, I.R. Petersen. \emph{Quantum Control Theory and applicatons: A survey}, IET Control Theory \& Applications, vol. 4, no. 12, pp.2651-2671 (2010)

%\bibitem [AD03]{AD} F. Albertini and D. D’Alessandro. {\em Notions of controllability for bilinear multilevel quantum systems} IEEE Trans. Autom. Control, 2003, 48, pp.1399-1403.


\bibitem[6]{DV1} V. Moretti and D. Pastorello. {\em Generalized Complex Spherical Harmonics,
Frame Functions, and Gleason Theorem}, Ann. Henri Poincar\'e 14,1435-1443
(2013).


\bibitem[7]{DV2}  V. Moretti and D. Pastorello. {\em  Frame functions in finite-dimensional Quantum Mechanics and its Hamiltonian formulation on complex projective spaces},
Int. J. Geom. Methods Mod. Phys. DOI: 10.1142/S0219887816500134 (2015)

\bibitem[8]{D1} D.Pastorello {\em Geometric Hamiltonian formulation of Quantum Mechanics in complex projective spaces},  International Journal of Geometric Methods in Modern Physics. DOI:10.1142/S0219887815600154 (2015)

\bibitem[9]{D}D. Pastorello. {\em A geometric Hamiltonian description of composite quantum systems and quantum entanglement}, Int. J. Geom. Methods Mod. Phys., 12, 1550069 (2015)



\bibitem[10]{book} H. Nijmeijer, A.J. van der Schaft {\em Nonlinear dynamical control systems}, \copyright   Springer-Verlag New York Inc. (1990) 



\bibitem[11]{SS} S.G. Schirmer, A.I. Solomon {\em Complete controllability of quantum system}, Physical rev. A 63, 063410 (2001)

\bibitem[12]{Wa} F.W. Warner {\em Foundations of Differentiable Manifolds and Lie Groups},
Springer, Berlin, 1983



\end{thebibliography}
\end{document}